\documentclass[12pt]{llncs}

\usepackage{times}
\usepackage{amsmath}
\usepackage{amsfonts}
\usepackage{amssymb}
\usepackage{graphicx,psfig,epsfig}
\usepackage{color}
\usepackage{algorithm}
\usepackage{algorithmic,xspace}
\usepackage{url}
\usepackage[colorlinks]{hyperref}
\usepackage{epsfig}

\usepackage{vmargin,fancyhdr}   
\setpapersize{USletter}
\setmarginsrb{.75in}{.5in}        
             {.75in}{.5in}        
             {.25in}{.25in}     
             {.25in}{.5in}      
\newcommand{\ben}{\begin{enumerate}}
\newcommand{\een}{\end{enumerate}}

\newcommand{\field}[1]{\mathbb{#1}} 

\newcommand{\beql}[1]{\begin{equation}\label{#1}}
\newcommand{\eeq}{\end{equation}}
\newcommand{\comment}[1]{}

\newcommand{\Abs}[1]{{\left|{#1}\right|}}

\newcommand{\Mean}[1]{{\mathbb E}\left[{#1}\right]}
\newcommand{\Var}[1]{{\mathbb Var}\left[{#1}\right]}

\newcommand{\Prob}[1]{{{\bf{Pr}}\left[{#1}\right]}}
\newcommand{\Set}[1]{{\left\{{#1}\right\}}}

\newcommand{\tr}{{\rm Tr\,}}

\begin{document}

\title{Efficient Triangle Counting in Large Graphs via Degree-based Vertex Partitioning}

\titlerunning{Efficient Triangle Counting in Graph Streams}  
%
\author{Mihail N. Kolountzakis\inst{1}, Gary L. Miller\inst{2} \and Richard Peng\inst{2} \and Charalampos E. Tsourakakis\inst{3}}
\authorrunning{Mihail Kolountzakis et al.}   
%
%
\tocauthor{Gary Miller, Richard Peng, Charalampos Tsourakakis}
\institute{ 
Department of Mathematics, University of Crete, Greece \\
\and
School of Computer Science, Carnegie Mellon University, USA\\
\email{kolount@math.uoc.gr}, \email{glmiller@cs.cmu.edu}, \email{yangp@cs.cmu.edu}, \email{ctsourak@math.cmu.edu} \\ 
\and
Deparment of Mathematical Sciences, Carnegie Mellon University, USA \\
 }

\maketitle

\begin{abstract}
The number of triangles is a computationally expensive graph statistic 
which is frequently used in complex network analysis (e.g., transitivity ratio),
in various random graph models (e.g., exponential random graph model) 
and in important real world applications such as spam detection,
uncovering of the hidden thematic structure of the Web and link recommendation. 
Counting triangles in graphs with millions and billions of edges requires algorithms
which run fast, use small amount of space, provide accurate estimates of the number of triangles 
and preferably are parallelizable. 

In this paper we present an efficient triangle counting algorithm which can be adapted 
to the semistreaming model \cite{feigenbaum}. The key idea of our algorithm is to combine the sampling 
algorithm of \cite{Tsourakakiskdd09,TsourakakisArxiv} and the partitioning of the
set of vertices into a high degree and a low degree subset respectively as in \cite{alon:alon},
treating each set appropriately. We obtain a running time  
$O \left( m + \frac{m^{3/2} \Delta \log{n} }{t \epsilon^2} \right)$
and an $\epsilon$ approximation (multiplicative error), where $n$ is the number of vertices,
$m$ the number of edges and $\Delta$ the maximum number of triangles an edge is contained. 
Furthermore, we show how this algorithm can be adapted to the semistreaming model 
with space usage $O\left(m^{1/2}\log{n} + \frac{m^{3/2} \Delta \log{n}}{t \epsilon^2} \right)$
and a constant number of passes (three) over the graph stream.
We apply our methods in various networks with several millions of edges and we obtain excellent results.
Finally, we propose a random projection based method for triangle counting and provide a sufficient 
condition to obtain an estimate with low variance. 

\end{abstract}

\section{Introduction}
\label{sec:intro}
Graphs are ubiquitous: the Internet, the World Wide Web (WWW), social networks, protein interaction networks 
and many other complicated structures are modeled as graphs \cite{chunglu}. 
The problem of counting subgraphs is one of the typical graph mining tasks that has attracted 
a lot of attention. The most basic, non-trivial subgraph, is the triangle. 
Given a simple, undirected graph $G(V,E)$, a triangle is 
a three node fully connected subgraph. Many social networks are abundant in triangles, since typically friends
of friends tend to become friends themselves \cite{faust:social}. This phenomenon is observed
in other types of networks as well (biological, online networks etc.) and is one of the main reasons which 
gave rise to the definitions of the transitivity ratio  and the clustering coefficients of a graph in
complex network analysis \cite{newman:structure}.
Triangles are used in several applications such as uncovering the hidden thematic structure of the web \cite{eckman:thematic},
as a feature to assist the classification of web activity \cite{gionis:spam} and for link recommendation 
in online social networks  \cite{Tsourakakisasonam09}. Furthermore, triangles are used as a network statistic in the exponential random graph model 
\cite{frank}.

In this paper, we propose a new triangle counting method which provides an $\epsilon$ approximation
to the number of triangles in the graph and runs in $O \left( m + \frac{m^{3/2} \Delta \log{n} }{t \epsilon^2} \right)$ time, 
where $n$ is the number of vertices, $m$ the number of edges and $\Delta$ the maximum number of triangles an edge is contained. 
The key idea of the method is to combine the sampling scheme introduced by Tsourakakis et al. in \cite{Tsourakakiskdd09,TsourakakisArxiv}
with the partitioning idea of Alon, Yuster and Zwick \cite{alon:alon} in order to obtain a more efficient sampling scheme. 
Furthermore, we show that this method can be adapted to the semistreaming model with a constant number of passes
and $O\left(m^{1/2}\log{n} + \frac{m^{3/2} \Delta \log{n} }{t \epsilon^2} \right)$ space. 
We apply our methods in various networks with several millions of edges and we obtain excellent results both with respect to the accuracy
and the running time. Furthermore, we optimize the cache properties of the code in order to obtain a significant additional speedup. 
Finally, we propose a random projection based method for triangle counting and provide a sufficient 
condition to obtain an estimate with low variance. 

The paper is organized as follows: Section~\ref{sec:prelim} presents briefly
the existing work and the theoretical background, Section~\ref{sec:method} 
presents our proposed method and Section~\ref{sec:experiments} presents
the experimental results on several large graphs. In Section~\ref{sec:ramifications} we 
provide a sufficient condition for obtaining a concentrated estimate 
of the number of triangles  using random projections and in Section~\ref{sec:concl} we conclude
and provide new research directions. 

\section{Preliminaries}
\label{sec:prelim}
In this section, we briefly present the existing work on the triangle counting problem 
and the necessary theoretical background for our analysis, namely a version of the Chernoff bounded
and the Johnson-Lindenstrauss lemma. Table \ref{tab:Symbols} lists the symbols used in this paper.

\subsection{Existing work}

There exist two categories of triangle counting algorithms, the exact and the approximate. 
It is worth noting that for the applications described in Section~\ref{sec:intro} the exact number of triangles in not crucial.
Thus, approximate counting algorithms which are faster and output a high quality estimate are desirable for the practical applications
in which we are interested in this work.

The state of the art algorithm is due to Alon, Yuster and Zwick \cite{alon:alon} 
and runs in $O(m^{\frac{2\omega}{\omega+1}})$,
where currently the fast matrix multiplication exponent $\omega$ is 2.371
\cite{CopperWino:CopperWino}. Thus, the Alon et al. algorithm currently runs in $O(m^{1.41})$ time. 
Algorithms based on matrix multiplication are not used in practice due to the high memory requirements.
Even for medium sized networks, matrix-multiplication based algorithms are not applicable. 
In planar graphs, triangles can be found in $O(n)$ time \cite{itai:rodeh,pap:yan}.
Furthermore, in \cite{itai:rodeh} an algorithm which finds a triangle in any graph in $O(m^\frac{3}{2})$
time is proposed. This algorithm can be extended to list the triangles in the graph with the same time complexity.
Even if listing algorithms solve a more general problem than the counting one, they
are preferred in practice for large graphs, due to the smaller memory requirements
compared to the matrix multiplication based algorithms. 
Simple representative algorithms are the node- and the edge-iterator algorithms.
The former counts for each node number of triangles it's involved in,
which is equivalent to the number of edges among its neighbors,
whereas in the latter, the algorithm counts for each edge $(i,j)$ the
common neighbors  of nodes $i,j$.
Both of these algorithms have the same asymptotic complexity $O(mn)$, 
which in dense graphs results in $O(n^3)$ time, the complexity of the naive counting algorithm. 
Practical improvements over this family of algorithms have been achieved using various techniques, such as 
hashing and sorting by the degree \cite{latapy,wagner:wagner}.

\begin{table}[t]
	\centering
	\begin{tabular}{c|p{0.35\textwidth}c|} \hline \hline
	
	Symbol & Definition \\ \hline \hline
        $G([n],E)$ & undirected simple graph with $n$ vertices labeled $1,2,..,n$  \\ 
                 &  and edge set $E$ \\ \hline
	$m$ & number of edges in $G$ \\ \hline
	$t$ & number of triangles in $G$ \\ \hline
        $deg(u)$ & degree of vertex $u$ \\ \hline
	$\Delta(u,v)$ & $\#$ triangles  \\ 
                      &  containing vertices $u$ and $v$ \\ \hline
	$\Delta$      &   $\max_{e \in E(G)}{\Delta(e)}$ \\ \hline
	$p$  & sparsification parameter \\ \hline \hline
\end{tabular}
	\caption{Table of symbols}
\label{tab:Symbols}
\end{table}

On the approximate counting side, most of the triangle counting algorithms have been developed in the 
streaming setting. In this scenario, the graph is represented as a stream. 
Two main representations of a graph as a stream are the edge stream and the incidence stream. In the former, edges are arriving
one at a time. In the latter scenario
all edges incident to the same vertex appear successively in the stream. The ordering of the vertices 
is assumed to be arbitrary. A streaming algorithm produces a relative $\epsilon$ approximation 
of the number of triangles with high probability, making a constant number of passes over the stream. 
However, sampling algorithms developed in the streaming literature can be applied in the setting where the graph fits in the memory as well. 
Monte Carlo sampling techniques have been proposed to give a fast estimate of the number of triangles.
According to such an approach, a.k.a. naive sampling \cite{shank:wanger1}, we choose three nodes at random repetitively and check if they form a triangle or not. 
If one makes $$ r = \log({\frac{1}{\delta}})\frac{1}{\epsilon^2}(1+\frac{T_0+T_1+T_2}{T_3})$$
independent trials where $T_i$ is the number of triples with $i$ edges 
and outputs as the estimate of triangles the random variable $T_3'$ equaling to the 
fractions of triples picked that form triangles times the total number of
triples (${n \choose 3}$), then
$$(1-\epsilon)T_3 < T_3' < (1+\epsilon)T_3 $$
with probability at least $1- \delta$.
This is not suitable when $T_3=o(n^2)$, which is often the case when dealing with
real-world networks.

In \cite{yosseff} the authors reduce the problem of triangle counting efficiently to estimating
moments for a stream of node triples. Then, they use the Alon-Matias-Szegedy algorithms \cite{amsalgos} (a.k.a. AMS algorithms) to proceed. 
The key is that the triangle computation reduces in estimating the zero-th, first and second frequency moments, which can be done efficiently. 
Again, as in the naive sampling, the denser the graph the better the approximation.
The AMS algorithms are also used by \cite{jowhary}, where simple sampling techniques are used, such
as choosing an edge from the stream at random and checking how many common neighbors its two endpoints share
considering the subsequent edges in the stream. 
Along the same lines,  \cite{buriol} proposed two space-bounded sampling algorithms to estimate the number of triangles. 
Again, the underlying sampling procedures are simple. E.g., for the case of the edge stream representation, they sample randomly
an edge and a node in the stream and check if they form a triangle. Their algorithms are the state-of-the-art algorithms to 
the best of our knowledge. The three-pass algorithm presented therein, counts  in the first pass the number of edges, in the second pass 
it samples uniformly at random an edge $(i,j)$ and a node $k \in V-\{i,j\}$ and in the third pass it tests
whether the edges $(i,k),(k,j)$ are present in the stream. The number of draws that have to be done in order to get 
concentration (these draws are done in parallel), is of the order
$$ r = \log({\frac{1}{\delta}})\frac{2}{\epsilon^2}(3+\frac{T_1+2T_2}{T_3})$$
Even if the term $T_0$ is missing compared to the naive sampling, the graph has still to be fairly dense with respect
to the number of triangles in order to get an $\epsilon$ approximation with high probability. 
In the case of ``power-law'' networks it was shown in \cite{me1} that the spectral counting of triangles 
can be efficient due to their special spectral properties and \cite{TsourakakisKAIS} extended this idea 
using the randomized algorithm by \cite{drineas:frieze} by proposing a simple biased node sampling. 
This algorithm can be viewed as a special case
of a streaming algorithm, since there exist algorithms, e.g., \cite{tamas}, that perform a constant number of passes
over the non-zero elements of the matrix to produce a good low rank matrix approximation.
In \cite{gionis:spam} the semi-streaming model for counting triangles is introduced, which allows $\log n$ passes over the edges. 
The key observation is that since counting triangles reduces to computing the intersection of two sets, namely the induced neighborhoods
of two adjacent nodes, ideas from  locality sensitivity hashing \cite{alan} are applicable to the problem. 
In \cite{Tsourakakiskdd09} an algorithm which tosses a coin independently for each edge with probability $p$ to keep the edge and probability $q=1-p$
to throw it away is proposed. It was shown later by Tsourakakis, Kolountzakis and Miller \cite{TsourakakisArxiv} 
using a powerful theorem due to Kim and Vu \cite{kim-vu} that under mild conditions on the triangle density the method results in a strongly concentrated
estimate on the number of  triangles. 
More recently, Avron proposed a new approximate triangle counting method based on a randomized algorithm
for trace estimation \cite{avron}. 

\subsection{Concentration of Measure}

In Section~\ref{sec:method} we make extensive use of the following version of the Chernoff bound \cite{chernoff}.

\begin{theorem} 
Let $X_1, X_2, \ldots, X_k$ be independently distributed $\{0,1\}$ variables with $E[X_i]=p$. Then for any $\epsilon > 0$, we have
$$Pr \left[ |\frac{1}{k} \sum_{i=1}^k X_i - p| > \epsilon p \right] \leq 2 e^{-\epsilon^2pk/2}$$
\end{theorem}

\subsection{Random Projections}

A random projecton $x \to Rx$ from $\field{R}^d \to \field{R}^k$ approximately preserves all Euclidean distances.
One version of the Johnson-Lindenstrauss lemma \cite{lindenstrauss} is the following:

\begin{lemma}[Johnson Lindenstrauss]
Suppose $x_1,\ldots,x_n \in \field{R}^d$ and $\epsilon>0$ and take $k=C \epsilon^{-2} \log n$. Define the random matrix $R \in \field{R}^{k\times n}$ by taking all $R_{i,j} \sim N(0,1)$ (standard gaussian) and independent. Then, with probability bounded below by a constant the points $y_j = R x_j \in \field{R}^k$ satisfy
$$
(1-\epsilon) \Abs{x_i-x_j} \le \Abs{y_i - y_j} \le (1+\epsilon) \Abs{x_i - x_j}
$$
for $i, j=1,2,\ldots,n$.
\label{thrm:jllemma}
\end{lemma}

\section{Proposed Method}
\label{sec:method}
Our algorithm combines two approaches that have been taken on triangle counting:
sparsify the graph by keeping a random subset of the edges \cite{Tsourakakiskdd09,TsourakakisArxiv}
followed by a triple sampling using the idea of vertex partitioning due to
Alon, Yuster and Zwick \cite{alon:alon}.

\subsection{Edge Sparsification}

The following method was introduced in  \cite{Tsourakakiskdd09} and was shown
to perform very well in practice: keep each edge with probability $p$
independently.
Then for each triangle, the probability of it being kept is $p^3$. So the
expected number of triangles left is $p^3t$.
This is an inexpensive way to reduce the size of the graph as it can be done in
one pass over the edge list using $O(mp)$ random variables (more details
can be found in section \ref{sec:sublin} and \cite{knuth}).

In a later analysis \cite{TsourakakisArxiv}, it was shown that the number of
triangles in the sampled graph is concentrated around the actual triangle count
as long as $p^3 \geq \tilde{\Omega}(\frac{\Delta}{t})$.
Here we show a similar bound using more elementary techniques.
Suppose we have a set of $k$ triangles such that no two share an edge,
for each such triangle we define a random variable $X_i$ which is $1$ if
the triangle is kept by the sampling and $0$ otherwise.
Then as the triangles do not have any edges in common,
the $X_i$s are independent and take value $0$ with probability $1-p^3$ and
$1$ with probability $p^3$.
So by Chernoff bound, the concentration is bounded by:

$$Pr \left[ |\frac{1}{k} \sum_{i=1}^k X_i - p^3| > \epsilon p^3 \right] \leq 2e^{-\epsilon^2p^3k/2}$$

So when $p^3 k \epsilon^2 \geq 4d\log{n}$, the probability of
sparsification returning an $\epsilon$-approximation is at least $1-n^{-d}$.
This is equivalent to $p^3 k \geq (4d\log{n}) /( \epsilon^2)$, so to sample with
small $p$ and throw out many edges, we would like $k$ to be large.
To show that such a large set of independent triangles exist, we invoke
the Hajnal-Szemer\'{e}di Theorem \cite{HajnalSzemeredi}:

\begin{lemma} (Hajnal-Szemer\'{e}di Theorem)
Every graph with $n$ vertices and maximum vertex degree at most $k$
is $k+1$ colorable with all color classes of size at least $n/k$.
\end{lemma}

We can apply this theorem by considering the graph where each triangle
is a vertex and two vertices representing triangles $t_1$ and $t_2$ are
connected iff they have an edge in common.
Then vertices in this graph has degree at most $O(\Delta)$, and we get:

\begin{corollary} \label{cor:partition}
Given $t$ triangles where no edge belongs to more than $\Delta$ triangles,
we can partition the triangles into $S_1 \dots S_l$ such that
$|S_i| > \Omega(t / \Delta)$ and $l$ is bounded by $O(\Delta)$.
\end{corollary}

We can now bound what values of $p$ can give concentration:

\begin{theorem}
If $p^3 \in \Omega(\frac{d \Delta\log{n}}{\epsilon^2 t})$, then with probability
$1-n^{d-3}$, the sampled graph has a triangle count that $\epsilon$-approximates
$t$.
\end{theorem}

\begin{proof}
Consider the partition of triangles given by corollary \ref{cor:partition}.
By choice of $p$ we get that the probability that the triangle count in
each set is preserved within a factor of $\epsilon/2$ is at least $1-n^{d}$.
Since there are at most $n^3$ such sets, an application of the union
bounds gives that their total is approximated within a factor
of $\epsilon/2$ with probability at least $1 - n^{d-3}$.
This gives that the triangle count is approximated
within a factor of $\epsilon$ with probability at least $1 - n^{d-3}$.
\end{proof} 

\subsection{Triple Sampling}

Since each triangle corresponds to a triple of vertices, we can construct a set of
 triples that include all triangles, $U$.
From this list, we can then sample some triples uniformly, let these samples be numbered from $1$ to $s$.
Also, for the $i^{th}$ triple sampled, let $X_i$ be $1$ it is a triangle and $0$
otherwise.
Since we pick triples randomly from $U$ and $t$ of them are triangles,
we have $E(X_i) = \frac{t}{|U|}$ and $X_i$s are independent.
So by Chernoff bound we obtain:

$$Pr \left[ |\frac{1}{s} \sum_{i=1}^s X_i - \frac{t}{|U|}| > \epsilon \frac{t}{|U|}
 \right] \leq 2e^{-\epsilon^2ts/(2|U|)}$$

So when $s = \Omega(|U|/t \log{n} / \epsilon^2)$, we have $(\frac{1}{s}
\sum_{i=1}^s X_i / s )|U|$
approximates $t$ within a factor of $\epsilon$ with probability at least $1-n^{-d}$
for any $d$ of our choice.
As $|U| \leq n^3$, this immediately gives an algorithm with runtime
$O(n^3\log{n}/(t\epsilon^2))$ that approximates $t$ within a factor
of $\epsilon$. Slightly more careful bookkeeping can also give tighter bounds on
$|U|$ in sparse graphs.

Consider a triple containing vertex $u$, $(u,v,w)$.
Since $uv, uw \in E$, we have the number of such triples involving $u$ is at
most $\text{deg}(u)^2$.
Also, as $vw \in E$, another bound on the number of such triples is $m$.
When $\text{deg}(u)^2 > m$, or $\text{deg}(u) > m^{1/2}$, the second
 bound is tighter, and the first is in the other case.

These two cases naturally suggest that low degree vertices with degree at most
$m^{1/2}$ be treated separately from high degree vertices with degree greater
than $m^{1/2}$.
For the number of triangles around low degree vertices, 
since $x^2$ is concave, the value of $\sum_u \text{deg}(u)^2$ is maximized
when all edges are concentrated in as few vertices as possible.
Since the maximum degree of such a vertex is $m^{1/2}$, the number of such
triangles is upper bounded by $m^{1/2} \cdot (m^{1/2})^2 = m^{3/2}$.
Also, as the sum of all degrees is $2m$, there can be at most $2m^{1/2}$ high
degree vertices, which means the total number of triangles incident to these
high degree vertices is at most $2m^{1/2} \cdot m = 2m^{3/2}$.
Combing these bounds give that $|U|$ can be upper bounded by $3m^{3/2}$.
Note that this bound is asymptotically tight when $G$ is a complete
graph ($n = m^{1/2}$).
However, in practice the second bound can be further reduced by
 summing over the degree of all $v$ adjacent to $u$, becoming 
$\sum_{uv \in E} \text{deg}(v)$.
As a result, an algorithm that implicitly constructs $U$ by picking the better one among these two cases by examining the degrees of all neighbors will achieve
$$|U| \leq O(m^{3/2})$$

This better bound on $U$ gives an algorithm that $\epsilon$ approximates the
number of triangles in time:

$$O \left( m + \frac{m^{3/2} \log{n}}{t \epsilon^2} \right)$$

As our experimental data in section 4.1. indicate, 
the value of $t$ is usually $\Omega(m)$ in practice.
In such cases, the second term in the above calculation becomes negligible
compared to the first one.
In fact, in most of our data, just sampling the first type of triples (aka. pretending
all vertices are of low degree) brings the second term below the first.

\subsection{Hybrid algorithm}

Edge sparsification with a probability of $p$ allows us to only work on $O(mp)$
edges, therefore the total runtime of the triple sampling algorithm after
sparsification with probability $p$ becomes:
$$O\left(mp+\frac{(mp)^{3/2} }{\epsilon^2 tp^3} \right) =
O\left(mp+\frac{m^{\alpha}}{\epsilon^2 t p^{3/2}} \right) $$

As stated above, since the first term in most practical cases are much larger,
we can set the value of $p$ to balance these two terms out:

\begin{align*}
pm &= \frac{m^{3/2}\log{n}}{p^{3/2}t\epsilon^2}\\
p^{5/2} t \epsilon^2 &= m^{1/2} \log{n}\\
p &= \left( \frac{m^{1/2}\log{n}}{t \epsilon^2} \right) ^{2/5}
\end{align*}

The actual value of $p$ picked would also depend heavily on constants in front
of both terms, as sampling is likely much less expensive due to factors such as
cache effect and memory efficiency.
Nevertheless, our experimental results in section 4 does seem to indicate that
this type of hybrid algorithms can perform better in certain situations.

\section{Experiments}
\label{sec:experiments}

\subsection{Data}

The graphs used in our experiments are shown in Table~\ref{tab:datasets}. 
Multiple edges and self loops were removed (if any).  

\begin{table}[ht]
\begin{center}
\begin{tabular}{l|r|r|r|l}
 Name  & Nodes & Edges &  Triangle Count & Description\\ \hline \hline

AS-Skitter & 1,696,415 & 11,095,298 & 28,769,868& Autonomous Systems \\ \hline
Flickr & 1,861,232  & 15,555,040 &  548,658,705 & Person to Person \\ \hline
Livejournal-links & 5,284,457 & 48,709,772 & 310,876,909 & Person to Person \\ \hline
Orkut-links & 3,072,626 & 116,586,585 & 285,730,264 & Person to Person \\ \hline
Soc-LiveJournal & 4,847,571 & 42,851,237 & 285,730,264 & Person to Person \\ \hline
Web-EDU&   9,845,725  &    46,236,104   & 254,718,147 &   Web Graph (page to page) \\ \hline
Web-Google & 875,713 & 3,852,985& 11,385,529 & Web Graph \\ \hline
Wikipedia 2005/11    & 1,634,989  & 18,540,589   & 44,667,095 & Web Graph    (page to page) \\\hline
Wikipedia 2006/9  & 2,983,494  & 35,048,115   &84,018,183& Web Graph  (page to page) \\ \hline
Wikipedia 2006/11 & 3,148,440  & 37,043,456 &88,823,817  & Web Graph (page to page) \\ \hline
Wikipedia 2007/2 & 3,566,907 & 42,375,911  & 102,434,918 & Web Graph (page to page) \\\hline
Youtube\cite{mislove} &   1,157,822  & 2,990,442&   4,945,382  & Person to Person \\ \hline
\end{tabular}
\end{center}
\caption{Datasets used in our experiments.}
\label{tab:datasets}
\end{table}

\subsection{Experimental Setup and Implementation Details} 

The experiments were performed on a single machine, with Intel Xeon CPU
at 2.83 GHz, 6144KB cache size and and 50GB of main memory. 
The graphs are from real world web-graphs, some details regarding them are in the chart below.
The algorithm as implemented in C++, and compiled using gcc version 4.1.2 and
the -O3 optimization flag.
Time was measured by taking the user time given by the linux time command.
IO times are included in that time since the amount of memory operations
performend in setting up the graph is non-trivial.
However, we use a modified IO routine that's much faster than the standard
C/C++ scanf.

A major optimization that we used was to sort the edges in the graph and store
the input file in the format as a sequence of neighbor lists per vertex.
Each neighbor list begins with the size of the list, followed by the neighbors.
This is similar to how softwares such as Matlab store sparse matrices, and
the preprocessing time to change the data into this format is not counted.
It can significantly improve the cache property of the graph stored, and therefore
improving the performance.

Some implementation details can be based on this graph storage format.
Since each triple that we check already have 2 edges already in the graph, it
suffices to check whether the 3rd edge in the graph.
This can be done offline by comparing a smaller list of edges against the initial
edge list of the graph and count the number of entries that they have in common.
Once we sort the query list, the entire process can be done offline in one pass
through the graph. This also means that instead of picking a pre-determined sample rate for the triples,
we can vary the sample rate for them so the number of queries is about the same
as the size of the graph. Finally, in the next section we discuss the details behind
efficient binomial sampling. Specifically picking a random subset of expected
size $p|S|$ from a set $S$ can be done in expected sublinear time \cite{knuth}. 

\subsubsection{Binomial Sampling in Expected Sublinear time}
\label{sec:sublin}

Most of our algorithms have the following routine in their core: given a list of values,
keep each of them with probability $p$ and discard with probability $1-p$.
If the list has length $n$, this can clearly be done using $n$ random variables.
As generating random variables can be expensive, it's preferrable to use
$O(np)$ random variables in expectation if possible.
One possibility is to pick $O(np)$ random elements, but this would likely
involve random accesses in the list, or maintaining a list of the indices picked
in sorted order. A simple way that we use in our code to perform this sampling is to generate the
differences between indices of entries retained \cite{knuth}.
This variable clearly belongs to an exponential distribution, and if $x$ is a uniform
random number in $(0, 1)$, taking $\lceil \log_{(1-p)}x \rceil$.
The primary advantage of doing so is that sampling can be done while accessing
the data in a sequential fashion, which results in much better cache performances.

\subsection{Results}

The six variants of the code involved in the experiment are first separated by
whether the graph was first sparsified by keeping each edge with probability $p = 0.1$.
In either case, an exact algorithm based on hybrid sampling with performance
bounded by $O(m^{3/2})$ is ran.
Then two triple based sampling algorithms are also considered.
They differ in whether an attempt to distinguish between low and high degree vertices,
so the simple version is essentially sampling all 'V' shaped triples off each
vertex. Note that no sparsification and exact also generates the exactly number of
triangles. Errors are measured by the absolute value of the difference between the value
produced and the exact number of triangles divided by the exact number.
The results on error and running time are averages over five runs. 
Results on these graphs described above are, the methods listed in the columns listed
in Table~\ref{tab:results}.

\begin{table}[ht]
\begin{center}
\begin{tabular}{p{2.5cm}|c|c|c|c|c|c|c|c|c|c|c|c|c}\hline \hline 
& \multicolumn{6}{|c|}{No Sparsification}
& \multicolumn{6}{|c|}{Sparsified ($p = .1$)} \\ \hline
Graph      & \multicolumn{2}{|c|}{Exact}
& \multicolumn{2}{|c|}{Simple}
&\multicolumn{2}{|c|}{Hybrid}
&\multicolumn{2}{|c|}{Exact}
& \multicolumn{2}{|c|}{Simple}
&\multicolumn{2}{|c|}{Hybrid} \\ \hline
& err(\%) & time & err(\%) & time & err(\%) & time & err(\%) & time & err(\%) & time & err(\%) & time \\ \hline
   
AS-Skitter
& 0.000 & 4.452
& 1.308 & 0.746
& 0.128 & 1.204
& 2.188 & 0.641
& 3.208 & 0.651
& 1.388 & 0.877
\\ \hline
Flickr
& 0.000 & 41.981
& 0.166 & 1.049
& 0.128 & 2.016
& 0.530 & 1.389
& 0.746 & 0.860
& 0.818 & 1.033
\\ \hline
Livejournal-links
& 0.000 & 50.828
& 0.309 & 2.998
& 0.116 & 9.375
& 0.242 & 3.900
& 0.628 & 2.518
& 1.011 & 3.475
\\ \hline
Orkut-links
& 0.000 & 202.012
& 0.564 & 6.208
& 0.286 & 21.328
& 0.172 & 9.881
& 1.980 & 5.322
& 0.761 & 7.227
\\ \hline
Soc-LiveJournal
& 0.000 & 38.271
& 0.285 & 2.619
& 0.108 & 7.451
& 0.681 & 3.493
& 0.830 & 2.222
& 0.462 & 2.962
\\ \hline
Web-EDU
& 0.000 & 8.502
& 0.157 & 2.631
& 0.047 & 3.300
& 0.571 & 2.864
& 0.771 & 2.354
& 0.383 & 2.732
\\ \hline
Web-Google
& 0.000 & 1.599
& 0.286 & 0.379
& 0.045 & 0.740
& 1.112 & 0.251
& 1.262 & 0.371
& 0.264 & 0.265
\\ \hline
Wiki-2005
& 0.000 & 32.472
& 0.976 & 1.197
& 0.318 & 3.613
& 1.249 & 1.529
& 7.498 & 1.025
& 0.695 & 1.313
\\ \hline
Wiki-2006/9
& 0.000 & 86.623
& 0.886 & 2.250
& 0.361 & 7.483
& 0.402 & 3.431
& 6.209 & 1.843
& 2.091 & 2.598
\\ \hline
Wiki-2006/11
& 0.000 & 96.114
& 1.915 & 2.362
& 0.530 & 7.972
& 0.634 & 3.578
& 4.050 & 1.947
& 0.950 & 2.778
\\ \hline
Wiki-2007
& 0.000 & 122.395
& 0.943 & 2.728
& 0.178 & 9.268
& 0.819 & 4.407
& 3.099 & 2.224
& 1.448 & 3.196
\\ \hline
Youtube
& 0.000 & 1.347
& 1.114 & 0.333
& 0.127 & 0.500
& 1.358 & 0.210
& 5.511 & 0.302
& 1.836 & 0.268
\\     \hline \hline 
\end{tabular}
\end{center}
\caption{Results of Experiments Averaged Over 5 Trials}
\label{tab:results}
\end{table}

\subsection{Remarks}

From Table~\ref{tab:results} it is clear that none of the variants clearly outperforms the others on
all the data. The gain/loss from sparsification are likely due to the fixed sampling rate, so
varying it as in earlier works \cite{Tsourakakiskdd09} are likely to mitigate
this discrepancy. The difference between simple and hybrid sampling are due to the fact
that handling the second case of triples has a much worse cache access
pattern as it examines vertices that are two hops away.
There are alternative implementations of how to handle this situation, which
would be interesting for future implementations.
A fixed sparsification rate of $p = 10\%$ was used mostly to simplify the setups
of the experiments. In practice varying $p$ to look for a rate where the result stabalizes is the preferred
option \cite{TsourakakisArxiv}.

When compared with previous results on this problem, the error rates and running
times of our results are all significantly lower.
In fact, on the wiki graphs our exact counting algorithms have about the same order 
of speed with other appoximate triangle counting implementations.

\section{Theoretical Ramifications} 
\label{sec:ramifications}
\subsection{Random Projections and Triangles}

Consider any two vertices $i,j \in V$ which are connected, i.e., $(i,j) \in E$. Observe that the inner product
of the $i$-th and $j$-th column of the adjacency matrix of graph $G$ gives the number of triangles that 
edge $(i,j)$ participates in. Viewing the adjacency matrix as a collection of $n$ points in $\field{R}^n$,
a natural question to ask is whether we can use results from the theory of random projections \cite{lindenstrauss}
to reduce the dimensionality of the points while preserving the inner products which contribute to the count
of triangles. Magen and Zouzias \cite{magen} have considered a similar problem, namely random projections which 
preserve approximately the volume for all subsets of at most $k$ points.

According to the  lemma~\ref{thrm:jllemma}, a random projecton $x \to Rx$ from $\field{R}^d \to \field{R}^k$ approximately preserves all Euclidean distances.
However it does not preserve all pairwise inner products. This can easily be seen by considering the set of points
$$
e_1, \ldots, e_n \in \field{R}^n = \field{R}^d.
$$
where $e_1=(1,0,\ldots,0)$ etc. 
Indeed, all inner products of the above set are zero, which cannot happen for the points $R e_j$ as they belong to a lower dimensional space 
and they cannot all be orthogonal. For the triangle counting problem we do not need 
to approximate {\em all} inner products. 
Suppose $A \in \Set{0,1}^n$ is the adjacency matrix of a simple undirected graph $G$ 
with vertex set $V(G) = \Set{1,2,\ldots,n}$ and write $A_i$ for the $i$-the column of $A$. 
The quantity we are interested in is the number of triangles in $G$ (actually six times the number of triangles)
$t = \sum_{u,v,w \in V(G)} A_{uv} A_{vw} A_{wu}$.

If we apply a random projection of the above kind to the columns of $A$
$A_i \to R A_i$ and write $X = \sum_{u,v,w \in V(G)} (RA)_{uv} (RA)_{vw} (RA)_{wu}$
it is easy to see that $\Mean{X}=0$ since $X$ is a linear combination of triple products $R_{ij} R_{kl} R_{rs}$ of entries of the 
random matrix $R$ and that all such products have expected value $0$, no matter what the indices. So we cannot expect this kind of random projection to work.

Therefore we consider the following approach which still has limitations as we will show in the following. Let
$t = \sum_{u \sim v} A_u^\top A_v,\ \ \mbox{where $u \sim v$ means $A_{uv}=1$},$
and look at the quantity
\begin{eqnarray*}
Y &=& \sum_{u \sim v} (R A_u)^\top (RA_v)\\
 &=& \sum_{l=1}^k \sum_{i,j=1}^n \left(\sum_{u \sim v} A_{iu} A_{jv} \right) R_{li} R_{lj}\\
 &=& \sum_{l=1}^k \sum_{i,j=1}^n \#\Set{i-*-*-j} R_{li} R_{lj}.
\end{eqnarray*}
This is a quadratic form in the gaussian $N(0,1)$ variables $R_{ij}$.
By simple calculation for the mean value and diagonalization for the variance we see that if the $X_j$ are independent $N(0,1)$ variables and
$$
Z = X^\top B X,
$$
where $X = (X_1,\ldots,X_n)^\top$ and $B \in \field{R}^{n \times n}$ is \textit{symmetric}, that
\begin{eqnarray*}
\Mean{Z} &=& \tr{B}\\
\Var{Z} &=& \tr{B^2} = \sum_{i,j=1}^n (B_{ij})^2.
\end{eqnarray*}
Hence $\Mean{Y} = \sum_{l=1}^k \sum_{i=1}^n \#\Set{i-*-*-i} = k\cdot t$ so the mean value is the quantity we want (multiplied by $k$).
For this to be useful we should have some concentration for $Y$ near $\Mean{Y}$. 
We do not need exponential tails because we have only one quantity to control. 
In particular, a statement of the following type
$$
\Prob{\Abs{Y-\Mean{Y}} > \epsilon \Mean{Y}} < 1-c_{\epsilon},
$$
where $c_{\epsilon}>0$ would be enough.
The simplest way to check this is by computing the standard deviation of $Y$. 
By Chebyshev's inequality it suffices that the standard deviation be much smaller than $\Mean{Y}$.
According to the formula above for the variance of a quadratic form we get
\begin{eqnarray*}
\Var{Y} &=& \sum_{l=1}^k \sum_{i,j=1}^n \#\Set{i-*-*-i}^2\\
 &=& C\cdot k \cdot \#\Set{x-*-*-*-*-*-x} =\\
 &=& C\cdot k \cdot \mbox{(number of circuits of length 6 in $G$)}.
\end{eqnarray*}
Therefore, to have concentration it is sufficient that
\beql{condition}
\Var{Y} = o(k \cdot (\Mean{Y})^2).
\eeq

Observe that \eqref{condition} is a sufficient -and not necessary-  condition.
Furthermore,\eqref{condition} is certainly not always true as there are graphs with many 6-circuits and no triangles at all (the circuits {\em may} repeat vertices or edges).

\subsection{Sampling in the Semi-Streaming Model}

The previous analysis of triangle counting by Alon, Yuster and Zwick was done
in the streaming model \cite{alon:alon}, where the assumption was constant
space overhead.
We show that our sampling algorithm can be done in a slightly weaker model
with space usage equaling:

$$O\left(m^{1/2}\log{n} + \frac{m^{3/2} \log{n}}{t \epsilon^2} \right)$$

We assume the edges adjacent to each vertex are given in order \cite{feigenbaum}.
We first need to identify high degree vertices, specifically the ones with degree
higher than $m^{1/2}$. This can be done by sampling $O(m^{1/2}\log{n})$ edges
and recording the vertices that are endpoints of one of those edges.

\begin{lemma}
Suppose $d m^{1/2} \log{n}$ samples were taken, then the probability of all
vertices with degree at least $m^{1/2}$ being chosen is at least $1-n^{-d+1}$.
\end{lemma}
\proof
Consider some vertex $v$ with degree at least $m^{1/2}$.
The probability of it being picked in each iteration is at least
$m^{1/2} / m = m^{-1/2}$.
As a result, the probability of it not picked in $d m^{1/2} \log{n}$ iterations is:
$$(1 - m^{-1/2})^{dm^{1/2}\log{n}}
= \left [(1-m^{1/2})^{m^{1/2}} \right ] ^ {d \log{n}}
\leq \left ( \frac{1}{e} \right ) ^ {d \log{n}} = n^{-d}$$
As there are at most $n$ vertices, applying union bound gives that all vertices
with degree at least $m^{1/2}$ are sampled with probability at least $1-n^{-d + 1}$.
\qed

This requires one pass of the graph. Note that the number of such candidates for
high degree vertices can be reduced to $m^{1/2}$ using another pass over the edge list.

For all the low degree vertices, we can read their $O(m^{1/2})$ neighbors
and sample them.
For the high degree vertices, we do the following: for each edge, obtain a random
variable $y$ from a binomial distribution equal to the number of
edge/vertices pairs that this edge is involved in.
Then pick $y$ vertices from the list of high degree vertices randomly.
These two sampling procedures can be done together in another pass over the data.

Finally, we need to check whether each edge in the sampled triples belong to the
edge list.
We can store all such queries into a hash table as there are at most
$O(\frac{m^{3/2} \log{n}}{t \epsilon^2})$ edges sampled w.h.p.
Then going through the graph edges in a single pass and looking them up in
table yields the desired answer.


\section{Conclusions \& Future Work}
\label{sec:concl}
In this work, we extended previous work \cite{Tsourakakiskdd09,TsourakakisArxiv} by introducing the powerful idea of Alon, Yuster and Zwick \cite{alon:alon}. 
Specifically, we propose a Monte Carlo algorithm which approximates the true number of triangles within $\epsilon$
and runs in $O \left( m + \frac{m^{3/2} \log{n} \Delta}{t \epsilon^2} \right)$ time. 
Our method can be extended to the semi-streaming model using three passes and a
memory overhead of $O\left(m^{1/2}\log{n} + \frac{m^{3/2} \log{n} \Delta}{t \epsilon^2} \right)$.

In practice our methods obtain excellent running times, typically few seconds for graphs with several millions of edges. The accuracy is also satisfactory,
especially for the type of applications we are concerned with. 
Finally, we propose a random projection based method for triangle counting and provide a sufficient 
condition to obtain an estimate with low variance. A natural question is the following:
can we provide some reasonable condition on $G$ that would guarantee \eqref{condition}? 
Finally, since our proposed methods are easily parallelizable, developing such an implementation in the \textsc{MapReduce} framework, see 
\cite{dean} and \cite{kangTF,kangTAFL}, is an natural practical direction.


%
%

\end{document}